\documentclass[a4paper, 11pt, onecolumn, draft]{quantumarticle}
\pdfoutput=1
\usepackage[utf8]{inputenc}
\usepackage[english]{babel}
\usepackage[T1]{fontenc}
\usepackage{amsmath}
\usepackage{empheq}
\usepackage{tikz}
\usepackage{lipsum}
\usepackage{array}
\usepackage{graphicx,enumerate}
\usepackage[colorlinks=true, allcolors=blue]{hyperref}
\usepackage{amssymb,mathtools,physics,amsthm}
\usepackage{biblatex}
\usepackage{csquotes}
\usepackage[T1]{fontenc}
\usepackage{fancyhdr}
\definecolor{quantumviolet}{HTML}{53257F} 
\newcommand{\changefont}{\fontsize{8.5}{11}\selectfont}

\fancypagestyle{alim}{\fancyfoot[C]{\textcolor{quantumviolet}{\thepage} \\ \changefont DISTRIBUTION STATEMENT A. Approved for public release; distribution is unlimited. OPSEC \#8804.}}
\newtheorem{prop}{Proposition}[section]

\def\diff{\mathrm{d}}

\title{An application of continuous-variable gate synthesis to quantum simulation of classical dynamics}

\author{Sam Cochran, James Stokes, Paramsothy Jayakumar and Shravan Veerapaneni}
\date{}

\begin{document}

\begin{abstract}
    Although quantum computing holds promise to accelerate a wide range of computational tasks, the quantum simulation of quantum dynamics as originally envisaged by Feynman remains the most promising candidate for achieving quantum advantage. A  less explored possibility with comparably far-reaching technological applicability is the quantum simulation of classical nonlinear dynamics. Attempts to develop digital quantum algorithms  based on the Koopman von Neumann formalism have met with challenges because of the necessary projection step from an infinite-dimensional Hilbert space to the finite-dimensional subspace described by a collection of qubits. This finitization produces numerical artifacts that limit solutions to very short time horizons. In this paper, we review continuous-variable quantum computing (CVQC), which naturally avoids such obstacles, and a CVQC algorithm for KvN simulation of classical nonlinear dynamics is advocated. In particular, we present explicit gate synthesis for product-formula Hamiltonian simulation of anharmonic vibrational dynamics.
\end{abstract}

\maketitle

\tableofcontents

\thispagestyle{alim}

\newpage

\section{Introduction}

The field of many-body system simulation presents a diverse landscape, in which the simulator $S$ and the target of simulation $T$ can consist of either classical (C) or quantum (Q) degrees of freedom, giving rise to a four-fold $S$-$T$ classification encompassing a wide array of simulation possibilities and associated complexities. Prominent examples include classical simulation of large-scale granular flow (C-C) \cite{de2019scalable} and classical simulation of quantum many-body dynamics (C-Q) using exact diagonalization techniques or compressed representations such as tensor network states \cite{orus2019tensor}. 

Recently, a shift of focus has occurred towards quantum simulation platforms, particularly those based on digital quantum information processing units (qubits). This shift is largely motivated by the rapid advancements and the near-term potential of quantum computers. 
The primary focus within this domain has been Q-Q, particularly simulation of discrete quantum dynamics, such as spin systems, using digital quantum simulation (DQS) techniques \cite{Daley2022}. In order to model fundamentally continuous problems such as electronic structure, a finite basis must necessarily be introduced which maps the problem to an effective spin model. 
Given the well-known discretization errors incurred by necessarily finite bases, an alternative strategy adapted to continuous-variable quantum systems is wanting. Continuous-variable quantum computing (CVQC) \cite{lloyd1999quantum} provides such an alternative by  leveraging continuous quantum information processing units (qumodes) and mapping each continuous degree of freedom in the system being simulated onto a corresponding qumode of the quantum device.

Despite the technological importance of classical simulation targets, most of the attention from quantum algorithm research  has been concentrated on Q-Q, and quantum simulation of classical systems (Q-C) has received comparatively little attention. The Koopman von Neumann formalism \cite{Koopman1931, JvN1, JvN2} provides a theoretical foundation for an approach to Q-C by effectively using quantum uncertainty to model the classical uncertainty of a stochastic classical dynamical system. The KvN formalism suggests that a quantum computer could be employed to simulate classical systems with an exponential and quadratic improvement in space and time complexity, respectively  \cite{kvn1}. This contrasts with quantum simulation of quantum systems (Q-Q), which achieve exponential improvement in both space and time complexity. DQS implementations of the KvN formalism necessarily involve domain discretization, which has been identified \cite{kvn2} as a limiting factor for long-time simulation. 
In particular, numerical simulations indicate that the under-resolved dynamics at sub-discretization scales induce Gibbs-like phenomena that prevent accurate simulation beyond short time horizons. Furthermore, Ref.~\cite{kvn2} noted that existing numerical methods for suppressing these artifacts tend to be nonlinear, making them incongruous with the inherently linear nature of quantum computing.

In this paper, we explore the application of CVQC to classical nonlinear dynamics, focusing on  Hamiltonian dynamical systems in which the Hamiltonian is a polynomial in the phase space variables up to degree four. In contrast to previous work, which focused on \emph{digital} quantum simulation of classical dynamical systems utilizing the Koopman von Neumann formalism \cite{kvn1, kvn2}, the CVQC techniques discussed in this paper have the potential to sidestep known limitations resulting from discretization artifacts, potentially achieving accurate simulation over longer time horizons than DQS.

In the existing CVQC literature, an explicit gate-based algorithm \cite{kvn3} has been provided to handle one-dimensional dynamical systems with polynomial vector fields of arbitrary degree. The present paper, in contrast, limits the interactions to degree four but considers Hamiltonian dynamical systems with arbitrary phase space dimensionality $\geq 2$. 
Thus, the dynamical systems considered here appear to be outside of the scope considered by Ref.~\cite{kvn3}. This paper thus  fills a gap in the literature by drawing on existing work on continuous-variable gate synthesis \cite{gate_decomposition2020, gate_decomposition2021} to describe product formula simulation of Hamiltonian dynamical systems suitable for modeling anharmonic vibrational dynamics.

This paper is structured as follows; in section \ref{sec:preliminaries} we review generalities of CVQC and the KvN formalism, which are necessary to describe the quantum algorithm, elaborated on in section \ref{sec:algo}. Section \ref{sec:example} discusses the example of the classical harmonic oscillator, and section \ref{sec:outlook} presents conclusions and future directions.

\section{Preliminaries}\label{sec:preliminaries}

\subsection{Continuous-variable quantum computing}\label{sec:CVQC}

The elementary units of quantum information processing in the CVQC framework are quantum modes, or qumodes. A qumode is a quantum-mechanical system whose Hilbert space can be identified with that of a quantum harmonic oscillator. The state of $n$ qumodes can thus be described in terms of the tensor product Hilbert space $L^2(\mathbb{R})^{\otimes n} \cong L^2(\mathbb{R}^n)$, and quantum gates are defined as unitary operations which are local with respect to this tensor product structure. In order to describe a gate set concretely, recall that the position and momentum quadrature operators for $j,k=1,\ldots,n$ satisfy the following defining canonical commutation relations
\begin{equation}
    [X_j, P_k] = i \delta_{jk},
\end{equation}
where $\delta_{jk}$ is the Kronecker delta. A set of gates is then said to be universal if it enables polynomial-depth approximation of unitary transformations $U_H(t) := \exp(-iHt)$, with $H$ a Hermitian polynomial in the quadrature operators of arbitrary but fixed degree \cite{lloyd1999quantum, cerezo2022challenges}.
The following is one example of a universal gate set \cite{gate_decomposition2021}:
\begin{subequations}
\begin{alignat}{2}
    &\text{Momentum displacement: } &\mathcal{D}_j(s) &= \exp\big( i s {X_j} \big) \label{eq:displacement}\\
    &\text{Quadratic phase: } &\mathcal{P}_j(s) &= \exp\big( is {X_j}^2 / 2 \big) \label{eq:quadratic_phase}\\
    &\text{Cubic phase: } &\mathcal{V}_j(s) &= \exp\big( is {X_j}^3 / 3 \big) \label{eq:cubic_phase}\\
    &\text{Rotation: } &\mathcal{R}_j(s) &= \exp \big( is \big( {X_j}^2 + {P_j}^2 \big) / 2 \big) \label{eq:rotation}\\
    &\text{Controlled-phase ($j\neq k$): } \qquad &\mathrm{CZ}_{jk}(s) &= \exp\big( i s X_j X_k \big). \label{eq:controlled-phase}
\end{alignat}
\end{subequations}

Recall that the Fourier transform rotates states between the position and momentum eigenbases. An important observation of the CVQC formulation is that the Fourier transform can be simply described in terms of the above gate set. Specifically, the Fourier transform on a given qumode $j$ can be implemented using a single rotation gate of angle $s = \pi / 2$. This gate is called the Fourier gate, and it can be used to efficiently transition between the position and momentum eigenbases:
\begin{subequations}
\begin{align}
    \mathcal{F}_j &:= \mathcal{R}_j(\pi / 2) \\
     X_j &=
    \mathcal{F}_j^\dagger P_j \mathcal{F}_j \\
    P_j &=
    -\mathcal{F}_j^\dagger X_j \mathcal{F}_j.
\end{align}
\end{subequations}
It follows that the controlled-X gate, defined for $j\neq k$ by
\begin{equation}
    \mathrm{CX}_{jk}(s) = \exp\big( {-i s X_j P_k} \big), \label{eq:controlled-x}
\end{equation}
can be obtained using a combination of the controlled-phase \eqref{eq:controlled-phase} and the Fourier gate.

\subsection{Hamiltonian dynamics and Koopman von Neumann theory}

In this section we consider a classical Hamiltonian dynamical system with even-dimensional phase space variable $x \in \mathbb{R}^{2n}$, Hamiltonian function $H=H(x,t)$, and symplectic matrix $J \in \mathbb{R}^{2n \times 2n}$ that satisfies $J^2 = -I_{2n}$,
\begin{equation}
    J =
    \begin{bmatrix}
    0 & I_n \\
    -I_n & 0
    \end{bmatrix}.
\end{equation}
If the initial state of the system is specified by the initial value of the phase space coordinate $x(0)=x_0 \in \mathbb{R}^{2n}$, then the time development of state variable $x=x(t)$ is determined  by the solution of Hamilton's equations
\begin{equation}
    \left\{
    \begin{aligned}
    \frac{\diff x}{\diff t}(t) & = J \frac{\partial H}{\partial x}(x(t),t), & t \in (0,\infty) \\
    x(0) & = x_0. & 
    \end{aligned}
    \right.
\end{equation}
If the initial state of the system is instead described by a probability density function $\rho_0 = \rho_0(x)$ on phase space, then the phase space density at subsequent times is given by the solution of Liouville's equation 
\begin{empheq}[left = \empheqlbrace]{align}\label{eq:ivp}
\begin{split}
    \frac{\partial \rho}{\partial t} (x,t) &= L [\rho](x,t), \qquad (x, t) \in \mathbb{R}^{2n} \times (0, \infty) \\ 
    \rho(x, 0) &= \rho_0(x), \qquad \qquad x \in \mathbb{R}^{2n}
\end{split}
\end{empheq}
where the partial differential operator $L$ is defined as
\begin{equation} \label{eq:liouvillian}
    L := \sum_{j=1}^n \bigg( \frac{\partial H}{\partial x_j} \frac{\partial}{\partial x_{n+j}} - \frac{\partial H}{\partial x_{n+j}} \frac{\partial}{\partial x_j} \bigg).
\end{equation}

At this point we remark that the following Koopman von Neumann (KvN) operator is formally self-adjoint:
\begin{equation}\label{eq:l_to_hkvn}
    H_\text{KvN} := i L.
\end{equation}
The operator $H_{\rm KvN}$ therefore defines a Hamiltonian for a quantum system, which acts on a suitable subspace of wave functions $\psi\in L^2(\mathbb{R}^{2n})$ defined over the classical phase space $x\in\mathbb{R}^{2n}$. The time evolution generated by $H_{\rm KvN}$ maps an initial wave function $\psi_0\in L^2(\mathbb{R}^{2n})$ to the solution of the time-dependent Schr\"{o}dinger equation
\begin{empheq}[left = \empheqlbrace]{align}\label{eq:ivp_schrodinger}
\begin{split}
    i \frac{\partial \psi}{\partial t} (x,t) &= H_\text{KvN} [\psi](x,t), \qquad (x, t) \in \mathbb{R}^{2n} \times (0, \infty) \\ 
    \psi(x, 0) &= \psi_0(x), \qquad \qquad \qquad x \in \mathbb{R}^{2n}.
\end{split}
\end{empheq}
Given a wave function $\psi \in L^2(\mathbb{R}^{2n})$, define the associated Born probability density as
\begin{equation}\label{eq:rho_def}
    \rho_\psi(x) := |\psi(x)|^2 \quad \quad  \textrm{for all} \quad \quad x \in \mathbb{R}^{2n}.
\end{equation}

KvN theory hinges on the following proposition, which motivates the CVQC algorithm presented in the next section.
\begin{prop}\label{prop}
Given  a probability density $\rho_0$ on phase space, let $\psi_0 \in L^2(\mathbb{R}^{2n})$ be any solution of $\rho_{\psi_0} = \rho_0$, and denote by $\psi$ the corresponding solution of the initial value problem \eqref{eq:ivp_schrodinger}. Then the function $\rho : \mathbb{R}^{2n}\times [0,\infty) \longrightarrow [0,\infty)$ defined by $\rho : (x,t) \longmapsto \rho_{\psi(\cdot,t)}(x)$ satisfies the initial value problem \eqref{eq:ivp}.
\end{prop}

\begin{proof}
Evaluating $\rho$ as defined above at $t=0$, we find that the initial condition is indeed satisfied:
$
    \rho(x, 0) := 
    \rho_{\psi(\cdot, 0)}(x) 
        = \rho_{\psi_0}(x) = \rho_0.
$
In order to show that $\rho = \psi^\ast\psi$ satisfies Liouville's equation for $t>0$, we recall  as a preliminary that the function $L[\cdot]$ satisfies the product rule on the relevant algebra of functions, which follows immediately from the definition \eqref{eq:liouvillian}. Proceeding to differentiate with respect to $t$ and using the fact that $\frac{\partial \psi}{\partial t} = L [\psi]$, the result follows as a consequence of elementary calculus:
\begin{align}
    \frac{\partial \rho}{\partial t} 
    & = \psi^* \frac{\partial \psi}{\partial t} + \frac{\partial \psi^\ast}{\partial t} \psi \\
    & = \psi^\ast L[\psi] + L[\psi^\ast] \psi \\
    & =L[\psi^\ast \psi] \\
    & = L[\rho].
\end{align}
\end{proof}

\section{Quantum algorithm}\label{sec:algo}

In the following section we walk through a quantum algorithm which, given a phase space distribution function $\rho_0$ and an evolution time $t \geq 0$, produces samples from the probability density function $\rho(\cdot, t)$ solving the initial value problem \eqref{eq:ivp}. KvN theory suggests the following approach:
\begin{enumerate}
    \item Prepare a system of $2n$ qumodes in a quantum state $\psi_0 \in L^2(\mathbb{R}^{2n})$ satisfying $\rho_{\psi_0} = \rho_0$.
    \item Evolve $\psi_0$ under the Hamiltonian $H_{\rm KvN}$ for evolution time $t$.
    \item Measure the resulting quantum state in the position quadrature basis.
\end{enumerate}
Note that the proof of correctness is essentially the content of Proposition \ref{prop}. Before proceeding, let us briefly discuss error bounds.  Commutator-based error bounds \cite{childs2021theory} are vacuous in our application because of the unboundedness of the KvN Hamiltonian \eqref{eq:l_to_hkvn}, which generically rules out uniform (state-independent) bounds. State-dependent error bounds for unbounded Hamiltonians are an active field of research \cite{burgarth2023strong}, the investigation of which we leave to future work.

\subsection{State preparation and measurement}

State preparation involves transforming the state of the quantum computer initially from from the vacuum configuration of $2n$ qumodes
\begin{equation}
    |0\rangle^{\otimes 2n} = \frac{1}{\pi^{n/2}}\int_{\mathbb{R}^{2n}} {\rm d}x \, e^{-\frac{1}{2}x^T x} |x\rangle
\end{equation}
to the desired state
\begin{equation}
    |\psi_0\rangle := \int_{\mathbb{R}^{2n}} {\rm d}x \, \psi_0(x) |x\rangle
\end{equation}
such that $\rho_{\psi_0}=\rho_0$. By the assumption of universality, any phase space distribution function of the form $\rho_0(x)=|\langle x | U_H(t)|0\rangle^{\otimes 2n}|^2$, for suitable $H$ discussed in section \ref{sec:CVQC}, can be efficiently approximated. This includes, as a special case, the set of Gaussian densities\footnote{See \cite{solovay-kitaev-extension} for the associated analysis of complexity.}.

After the time evolution step has been carried out, the measurement step simply involves estimating the target distribution $\rho(\cdot,t)$ from samples generated by the measurement outcomes obtained from projective measurements in the position quadrature basis. 

\subsection{Time evolution step}\label{sec:prob_form}

After the initial state has been prepared, the time evolution step involves preparing the state 
\begin{equation}
    |\psi(\cdot,t)\rangle := \int_{\mathbb{R}^{2n}} {\rm d}x\, \psi(x,t) |x\rangle
\end{equation}
from the initial state $|\psi_0\rangle$.  In this step, we pursue a simulation strategy using product formulas, which are based on identifying a decomposition of a target Hamiltonian $G$ as a linear combination of generators
\begin{equation} \label{eq:ham_lin_comb}
    G = \sum_{a=1}^L G_a
\end{equation}
such that any operator in the set $\{U_{G_a}(t) : t \in \mathbb{R}, \; a=1,\ldots,L \}$ admits an efficient implementation. Denote by $\mathcal{G}$ the subgroup of unitary operators generated by the efficiently simulable set.
A product formula is defined as a function $S:\mathbb{R} \longrightarrow \mathcal{G}$ such that $S(t)$ approximates $U_G(t)$ in the sense that $S(t/n)^n \approx U_G(t)$ for $n \geq 1$ sufficiently large. The simplest example is the first order Trotter-Suzuki product formula, which is defined as
\begin{equation}
    S_1(t) := \prod_{a=1}^L U_{G_a}(t).
\end{equation}

Product formulas are central to fault-tolerant quantum simulation algorithms, as well as variational quantum algorithms for time evolution \cite{otten2019noise, lin2021real, berthusen2022quantum, barison2021efficient}. In this work, we take an agnostic position regarding the choice of simulation algorithm, and instead focus on the identification and implementation of the efficiently simulable generating set $\{U_{G_a}(t) : t \in \mathbb{R}, \; a=1,\ldots,L \}$ for $G=H_{\rm KvN}$, from which all product formulas can be efficiently derived. 
Recalling that the canonical commutation relations are realized on $L^2(\mathbb{R}^{2n})$ by identifying $P_j = \frac{1}{i}\frac{\partial}{\partial X_j}$, the target Hamiltonian \eqref{eq:l_to_hkvn} is found to take the following form:
\begin{equation}\label{eq:kvn_ham}
    H_\text{KvN} = 
    \sum_{j=1}^n \bigg(\frac{\partial H}{\partial X_{n+j}} P_j - \frac{\partial H}{\partial X_j} P_{n+j} \bigg).
\end{equation}
In order to simplify the analysis\footnote{Relaxing this assumption would lead to a conceptually similar albeit more complicated analysis involving the approximate gate decomposition techniques discussed in Ref.~\cite{gate_decomposition2021}.}, we will ignore position-momentum cross terms in the classical Hamiltonian $H$, assuming $H$ can be written as
\begin{equation}\label{eq:no_PX_overlap}
    H = V(x_1, \dots, x_n) + T(x_{n + 1}, \dots, x_{2n}).
\end{equation}
The KvN Hamiltonian thus further simplifies to 
\begin{equation}\label{eq:kvn_ham_assumptions}
    H_\text{KvN} = 
    \sum_{j=1}^n \bigg(\frac{\partial T}{\partial X_{n+j}} P_j - \frac{\partial V}{\partial X_j} P_{n+j} \bigg).
\end{equation}

In order to draw on existing gate synthesis results \cite{gate_decomposition2020,gate_decomposition2021}, we restrict $T$ and $
V$ to be quartic polynomials, in which case $H_{\rm KvN}$ can be expressed in terms of three real-valued tensors $\alpha_k$, $\alpha_{kl}$, $\alpha_{klm}$ as follows:
\begin{align}\label{eq:kvn_ham_expanded}
\begin{split}
    H_\text{KvN} = \sum_{j=1}^n P_j \bigg( \sum_{k=n+1}^{2n} \alpha_k X_k 
    + \sum_{k, l=n+1}^{2n} \alpha_{kl} X_k X_l
    + \sum_{k,l,m=n+1}^{2n} \alpha_{klm} X_k X_l X_m \bigg) \\
    + \sum_{j=n+1}^{2n} P_j \bigg( \sum_{k=1}^n \alpha_k X_k 
    + \sum_{k,l=1}^n \alpha_{kl} X_k X_l
    + \sum_{k,l,m=1}^n \alpha_{klm} X_k X_l X_m \bigg).
\end{split}
\end{align}
We note that \eqref{eq:kvn_ham_expanded} uses the separation of variables assumption \eqref{eq:no_PX_overlap}, which implies that $[P_j,X_k]=0$ within each term of the Hamiltonian.

From the above form of $H_{\rm KvN}$, we readily identify an efficiently simulable set that consists of six kinds of time evolution operators, listed by polynomial order of the associated generator in Table \ref{tab:generators}.

\renewcommand{\arraystretch}{1.5}
\begin{table}
\begin{center}
\begin{tabular}{  c | c | c  } 
  
  \textbf{Quadratic} & \textbf{Cubic} & \textbf{Quartic} \\
  \hline
  $ \qquad\exp \big({-i s P_1 X_2} \big)\qquad$
  & $\exp \big({-i s P_1 {X_2}^2} \big)$ 
  & $\exp \big({-i s P_1 {X_2}^3} \big)$ \\ 
  
  & $\qquad \exp \big({-i s P_1 X_2 X_3} \big) \qquad$ 
  & $\exp \big({-i s P_1 {X_2}^2 X_3} \big)$ \\ 
  
  && $\qquad \exp \big({-i s P_1 X_2 X_3 X_4} \big) \qquad$ \\ 
  
\end{tabular}
\caption{In order to implement time evolution under the KvN Hamiltonian using product formulas, gate decompositions for six local unitary operators are required. These six operators are shown here, grouped by the polynomial degree of the quadrature operators in the exponent. Each operator involves between two and four qumodes, and for convenience, we write each operator here in terms of the first four qumode indices \{1, 2, 3, 4\}. These operations on different qumodes can be implemented in the same way by simply swapping the indices. Because quadrature operators on different qumodes always commute, the operations given in this table are sufficient to implement any KvN Hamiltonian satisfying the assumptions of section \ref{sec:prob_form}. \label{tab:generators}}
\end{center}
\end{table}

\subsection{Gate decomposition}

In this section, we summarize the literature leading to exact gate decompositions of the simulable set in table \ref{tab:generators}. To begin, observe that each operation has the form
\begin{equation}
    \exp\big({-i s P_1 {X_2}^{a_2} {X_3}^{a_3} {X_4}^{a_4}}\big),
\end{equation}
where the $a_j$ are nonnegative integers with $1 + a_2 + a_3 + a_4 = a$ giving the polynomial degree of the generator. Note that $a=2$ corresponds to a controlled-X gate, so it sufficies to consider $a\in\{3,4\}$. Using a Fourier transform on the first mode, we can transform the momentum quadrature into a position quadrature
\begin{equation} \label{eq:general_quartic}
    \exp\big({-i s P_1 {X_2}^{a_2} {X_3}^{a_3} {X_4}^{a_4}}\big) = 
    \mathcal{F}_1^\dagger 
    \exp\big(i s X_1 {X_2}^{a_2} {X_3}^{a_3} {X_4}^{a_4}\big)
    \mathcal{F}_1.
\end{equation}

We now apply the procedure described in section 4.3 of ~Ref. \cite{gate_decomposition2021} to expand the exponent in the right hand side of \eqref{eq:general_quartic} into a form that admits a decomposition into elementary gates. First, we use the identity
\begin{equation}\label{eq:expansion}
    X_1 {X_2}^{a_2} {X_3}^{a_3} {X_4}^{a_4} = 
    \sum_{v_2=0}^{a_2}
    \sum_{v_3=0}^{a_3}
    \sum_{v_4=0}^{a_4}
    C(v)
    \bigg(\sum_{i = 1}^4 h_i X_i\bigg)^a,
\end{equation}
where $h_1 = 1$, $h_i = a_i - 2 v_i$ for $i \in \{2, 3, 4\}$, and $C(v)$ is defined as 
\begin{equation}
    C(v) = \bigg( \frac{1}{2^{a - 1} a!} (-1)^{\sum_{i = 2}^4 v_i} \bigg) 
\binom{a_2}{v_2} \binom{a_3}{v_3} \binom{a_4}{v_4} .
\end{equation}
All terms in the sum in \eqref{eq:expansion} commute, so the exponential in \eqref{eq:general_quartic} can be written as
\begin{equation}\label{eq:exp_decomp}
    \exp\big({i s X_1 {X_2}^{a_2} {X_3}^{a_3} {X_4}^{a_4}}\big) = 
    \prod_{v_2 = 0}^{a_2} 
    \prod_{v_3 = 0}^{a_3} 
    \prod_{v_4 = 0}^{a_4}
    \exp\Bigg(i s C(v) \bigg(\sum_{i = 1}^4 h_i X_i\bigg)^a\Bigg).
\end{equation}
Implementing \eqref{eq:general_quartic} in terms of elementary gates thus amounts to repeated application of operators of the form appearing in the above iterated product \eqref{eq:exp_decomp}.
~Ref. \cite{gate_decomposition2021} has shown that the multiplicand in \eqref{eq:exp_decomp} can be decomposed as 
\begin{equation}\label{eq:key_decomp} 
    \exp\Bigg(i s C(v) \bigg(\sum_{i = 1}^4 h_i X_i\bigg)^a\Bigg)
    = U^\dagger
    \exp\big({i s C(v) {X_1}^a}\big)
    U,
\end{equation}
where 
\begin{equation}
    U := \exp \big( {-ih_2 X_2 P_1} \big) 
    \exp \big( {-ih_3 X_3 P_1} \big) 
    \exp \big( {-ih_4 X_4 P_1} \big).
\end{equation}

The detailed proof of \eqref{eq:key_decomp} is reviewed in appendix \ref{appendix:decomp_proof}. Notice that $U$
can be implemented using the controlled-X gates defined in \eqref{eq:controlled-x}. We now discuss the implementation of $\exp\big({i s C(v) {X_1}^a}\big)$. 
In the specific case of $a = 3$, this amounts to the cubic phase gate defined in \eqref{eq:cubic_phase}. In the case of $a = 4$, a decomposition for $\exp\big(i s C(v) {X_1}^4\big)$ has been provided in ~Ref. \cite[Eqs.~(56-58)]{gate_decomposition2021}.

In the specific case of $a = 3$, more efficient decompositions exist. Specifically, a decomposition for the first cubic entry in table \ref{tab:generators} is provided in ~Ref. \cite[Eq.~(51)]{gate_decomposition2021}. A gate decomposition for $e^{i s X_1 X_2 X_3}$ is provided in ~Ref. \cite[Eq.~(53)]{gate_decomposition2021}. This decomposition can be used for the second cubic entry in table \ref{tab:generators} after taking a Fourier transform on the mode corresponding to the momentum quadrature
\begin{equation}
    \exp \Big( {-i s P_1 X_2 X_3} \Big) = \mathcal{F}_1^\dagger \exp \Big(i s X_1 X_2 X_3 \Big) \mathcal{F}_1.
\end{equation}

\section{Harmonic oscillator example}\label{sec:example}

It is instructive to consider the example of the simple harmonic oscillator, for which the classical Hamiltonian is given by
\begin{equation}\label{eq:ham}
    H = \frac{1}{2m}  {x_2}^2 + \frac{1}{2} m \omega^2  {x_1}^2,
\end{equation}
where $x_1$ and $x_2$ denote the position and  momentum, respectively.  The Liouvillian and the associated KvN Hamiltonian for this system are thus given by 
\begin{align}
    L & = m \omega^2 x_1 \frac{\partial}{\partial  x_2} - \frac{1}{m} x_2 \frac{\partial}{\partial x_1}, \\
    H_\text{KvN} & = \frac{1}{m} X_2 P_1 - m \omega^2 X_1 P_2.
\end{align}
It follows by inspection of $H_{\rm KvN}$ that the only gate required for product formula simulation is the controlled-X \eqref{eq:controlled-x}.

\section{Outlook}\label{sec:outlook}

The KvN formalism potentially offers the opportunity for quantum acceleration of classical dynamics. Previous work \cite{kvn2} examined KvN dynamics through the lens of digital quantum computation and identified obstacles to achieving reliable solutions due to necessary finitization of the continuous dynamics. 

It is plausible that this shortcoming can be entirely circumvented in the CVQC framework. The use of an inherently continuous computing platform to simulate continuous systems avoids the issue of discretization errors, which could potentially enable simulation over longer time horizons compared to DQS. In this paper, we provided a blueprint for mapping KvN systems onto a continuous-variable quantum computer, using the example of vibrational dynamics as motivation. 

A number of directions for future work suggest themselves. The first concerns extending gate synthesis techniques to allow for higher-order polynomial interactions. The method provided by Ref.~\cite{gate_decomposition2021} only allows for exact decomposition of unitaries whose associated generator has a polynomial order divisble by either two or three. In order to implement unitaries whose generator has, for example, a polynomial degree of five or seven, new decompositions would need to be derived. Second, approximate decomposition methods \cite{gate_decomposition2021} could be used to implement unitaries that involve phase space cross-terms, allowing for simulation of Hamiltonian systems that do not admit the separation of variables \eqref{eq:no_PX_overlap}. Finally, it would be interesting to investigate state-dependent error bounds for time evolution generated by unbounded Hamiltonians of KvN form.

\section{Acknowledgements}
We acknowledge support from
the Automotive Research Center at the University of Michigan (UM) in accordance with Cooperative Agreement W56HZV-19-2-0001 with U.S. Army DEVCOM Ground Vehicle Systems Center.

\printbibliography

\appendix

\section{Proof of decomposition}\label{appendix:decomp_proof}

We now outline a proof for \eqref{eq:key_decomp}. It suffices to show that 
\begin{equation}\label{eq:wts}
    \exp\big(i s (X_1 + X_2)^a\big) = \exp(i P_1 X_2) \exp(i s {X_1}^a) \exp(-i P_1 X_2),
\end{equation}
as \eqref{eq:key_decomp} can then be obtained by repeated application of this identity. Although this result is well known, we provide the full details here for the convenience of the reader. Recall the well-known identities
\begin{align}\label{eq:unitary_conjugation}
    U^\dagger e^{i s H} U
    & = e^{i s U^\dagger H U} \\
    e^{A}Be^{-A} &  = B + [A, B] + \frac{1}{2!} [A, [A, B]] + \frac{1}{3!} [A, [A, [A, B]]] + \cdots \label{eq:bch_lemma}
\end{align}
as well as the following commutator properties:
\begin{align}
    [A, BC] &=[A, B]C+B[A, C] \\
    [AB, C] & = A[B, C] + [A, C]B \\
    [A, B^a] & = \sum_{j = 0}^{a - 1} B^j [A, B] B^{a - j - 1}.
\end{align}
Applying the identity \eqref{eq:unitary_conjugation} to the right hand side of \eqref{eq:wts}, we have that
\begin{align}
    e^{i P_1 X_2} e^{i s {X_1}^a} e^{-i P_1 X_2} = 
    \exp(i s e^{i P_1 X_2} {X_1}^a e^{-i P_1 X_2}).
\end{align}
The terms in the exponent can then be expanded using \eqref{eq:bch_lemma} to arrive at
\begin{equation}
    i s e^{i P_1 X_2} {X_1}^a e^{-i P_1 X_2} = i s \Big( {X_1}^a 
    + [i P_1 X_2, {X_1}^a] 
    + \frac{1}{2!}[i P_1 X_2, [i P_1 X_2, {X_1}^a]] + \dotsc \Big). 
\end{equation}
Next, the the above commutator properties can be used to rewrite each commutator term in the sum. The commutator in the second term becomes
\begin{align}
    [i P_1 X_2, {X_1}^a] 
    &= [i P_1, {X_1}^a]X_2 \\
    &= i\sum_{j = 0}^{a - 1} {X_1}^j [P_1, X_1] {X_1}^{a - j - 1} X_2 \\
    &= -i^2\sum_{j = 0}^{a - 1} {X_1}^{a - 1} X_2 \\
    &= a {X_1}^{a - 1} X_2.
\end{align}
Similarly, the commutator in the third term can be rewritten as 
\begin{align}
    [i P_1 X_2, [i P_1 X_2, {X_1}^a]] 
    &= [i P_1 X_2, a {X_1}^{a - 1} X_2] \\
    &= [i P_1, a {X_1}^{a - 1} X_2] X_2 \\
    &= [i P_1, a {X_1}^{a - 1}]{X_2}^2 \\
    &= i a\sum_{j = 0}^{a - 2} {X_1}^j [P_1, X_1] {X_1}^{a - j - 2} {X_2}^2 \\ 
    &= -i^2 a \sum_{j = 0}^{a - 2} {X_1}^{a - 2} {X_2}^2 \\
    &= a (a - 1) {X_1}^{a - 2} {X_2}^2.
\end{align}
The same approach shows that the commutator in the fourth term becomes
\begin{equation}
    [i P_1 X_2, [i P_1 X_2, [i P_1 X_2, {X_1}^a]]] = a (a - 1) (a - 2) {X_1}^{a - 3} {X_2}^3,
\end{equation}
and indeed the pattern continues all the way up to the $(a + 1)$th commutator term, which becomes 
\begin{equation}
    \underbrace{[i P_1 X_2, \dotsc  [i P_1 X_2, [i P_1 X_2}_{a \text{ times}}, {X_1}^a]] \dotsc ]
    = a! {X_2}^a.
\end{equation}
All commutators in subsequent terms are zero. We therefore have that 
\begin{align}
    &i s \Big( {X_1}^a 
    + [i P_1 X_2, {X_1}^a] 
    + \frac{1}{2!}[i P_1 X_2, [i P_1 X_2, {X_1}^a]] 
    + \dotsc \Big) \\
    &= is \Big( {X_1}^a + a {X_1}^{a - 1} X_2
    + \frac{a (a - 1)}{2!} {X_1}^{a - 2} {X_2}^2 \\
    &\quad+ \frac{a (a - 1) (a - 2)}{3!}{X_1}^{a - 3} {X_2}^3
    + \dotsc  + \frac{a!}{a!}{X_2}^a \Big) \\
    &= is(X_1 + X_2)^a,
\end{align}
and thus it is shown that 
\begin{equation}
    e^{i P_1 X_2} e^{i s {X_1}^a} e^{-i P_1 X_2} = 
    e^{i s (X_1 + X_2)^a}.
\end{equation}

\end{document}